\documentclass[letterpaper, 10 pt, onecolumn]{ieeeconf}  

\IEEEoverridecommandlockouts                              
\DeclareMathAlphabet{\mathpzc}{OT1}{pzc}{m}{it}
\usepackage{epsfig} 
\usepackage{times} 
\usepackage{amsmath} 
\usepackage{amssymb}  
\usepackage{setspace}
\usepackage{color}
\begin{document}
\title{The Kirchhoff-Braess Paradox and Its Implications for Smart Microgrids}
\author{John Baillieul, Bowen Zhang, Shuai Wang}

\newtheorem{thm}{Theorem}
\newtheorem{theorem}{Theorem}
\newtheorem{proposition}{Proposition}
\newtheorem{lemma}{Lemma}
\newtheorem{corollary}{Corollary}

\newtheorem{definition}{Definition}
\newtheorem{dfn}{Definition}
\newtheorem{remark}{Remark}
\newtheorem{prob}{Problem}
\newtheorem{example}{Example}
\newtheorem{cond}{Condition}
\graphicspath{{../../../Images/}}

\thispagestyle{empty}

\IEEEaftertitletext{\vspace{-2\baselineskip}} 
\maketitle
\thispagestyle{empty}

 \let\thefootnote\relax\footnotetext{John Baillieul is with the Departments of Mechanical Engineering; Electrical and Computer Engineering, and the Division of Systems Engineering.  Bowen Zhang and Shuai Wang are  with the Division of Systems Engineering at Boston University, Boston, MA 02115. Corresponding author is John Baillieul (Email: johnb@bu.edu). \newline This work was supported by the U.S. National Science Foundation under EFRI Grant Number 1038230.}

\begin{abstract}
Well known in the theory of network flows, Braess paradox states that in a congested network, it may happen that adding a new path between destinations can increase the level of congestion.  In transportation networks the phenomenon results from the decisions of network participants who selfishly seek to optimize their own performance metrics.  In an electric power distribution network, an analogous increase in congestion can arise as a consequence Kirchhoff's laws.  Even for the simplest linear network of resistors and voltage sources, the sudden appearance of congestion due to an additional conductive line is a nonlinear phenomenon that results in a discontinuous change in the network state.  It is argued that the phenomenon can occur in almost any grid in which they are loops, and with the increasing penetration of small-scale distributed generation, it suggests challenges ahead in the operation of microgrids.
\end{abstract}


\section{Introduction}

A good deal of current research on the operation of {\em smart grids} has been focused on the information structures and protocols that enable operation, \cite{Zhang2012},\cite{Zhang2013a},\cite{Zhang2013b},\cite{Zhang2014a},\cite{Zhang2014b}.  To facilitate load management at the level of a residential or commercial building microgrid, the authors have proposed the concept of packetized direct load control (PDLC).  Here the term packetized refers to a
fixed time energy usage authorization, with the emphasis being
on usage by thermostatic appliances (water heaters,
refrigerators, air conditioners, etc.) that have been aggregated into
pools defined by common operating characteristics. Taking HVAC, for example, 
consumers in each
room in a building choose their preferred set point, and then an operator ({\em smart building operator}, SBO) 
of the local appliance pool will determine an appropriate comfort band
around the set point. Our previous work has
shown that PDLC is capable of ensuring consumer comfort while at
the same time reducing power oscillations that occur when no control is
applied and appliances operate independently, \cite{Zhang2012} .  

While packet-switched energy distribution protocols appear to have numerous advantages in controlling demand oscillations from predictable loads, open questions remain concerning the operation of small-scale distribution networks in which multiple distributed generation (DG) sources inject power with hard-to-predict intermittency.  In such settings, stable and secure grid operation will be increasingly challenged by non-radial distribution network topologies in which lines connecting loads and power sources are opened and closed according to arrays of factors including weather, time-of-day, and real-time energy market conditions.

Focusing the discussion on network topology control, Section II begins with a review of {\em Braess paradox} in transportation networks.  Section III focuses on an analogous phenomenon in a voltage controlled circuit.  Here it is shown that changing the circuit topology by adding a small load can lead to relatively large losses in the circuit as a whole.  Section IV discusses similar sensitivity in optimal power flows within a simple non-radial distribution grid.  Section V extends the analysis of Section III to voltage-controlled circuits of arbitrary size.  Open problems are discussed in Section VI.

\section{Problem Formulation}

The principal challenge addressed here is the management of network congestion in the presence of load and DG (distributed generation) uncertainties.  We begin by recalling the well known Braess paradox that is generally associated with congestion in transportation networks.  The setup, in simple form, is shown in Fig.\ \ref{fig:jb:Braess}.  There is a network (of roads) with an origin {\sl O} and destination {\sl D}.  A certain number of travelers will make the journey, and in the  network in Fig.\ \ref{fig:jb:Braess}(A) they have a choice of the route with segments $\Large\mathpzc{AB}$ or segments $\mathpzc{ CD}$.  Congestion may enter either route in terms of travel times that depend on the number of users traveling on each segment.  If $f$ denotes the number of voyagers on the segment, the travel times on segments ${\mathpzc A}$ and ${\mathpzc D}$ are the same and equal to $f+\beta$ for some constant $\beta$.  The travel times on segments ${\mathpzc B}$ and ${\mathpzc C}$ are similarly equal to $\alpha f$ for some positive $\alpha$.  There are many different values of the parameters used in the literature, but the basic idea is that because the left-hand and right-hand routes in Fig.\ \ref{fig:jb:Braess} have the same congestion cost, $(\alpha+1)f+\beta$, introducing the cross link will break the cost symmetry and could cause the cost of travel to increase.  Taking the particular values of \cite{Steinberg}: $\alpha=10,\beta=50,\gamma(\cdot)\equiv 0$, and letting the total number of travelers be $f=6$, we find that without the cross link, the best choice for minimizing travel time is for three of the travelers to choose the left-hand route and three to choose the right hand route.  This is a Nash equilibrium.  The travel time for each traveler is $(\alpha + 1)f+\beta= 11\cdot 3+50=83$.  When the no-cost cross link is present,however, travelers will observe a possibly shorter route given by following the segments ${\mathpzc C}$-cross-link-${\mathpzc B}$.  Indeed, if only three of the six travelers took this route, the travel cost could be as low as $2\alpha f=60$.  Unfortunately, all six travelers may choose the route, in which case, the cost becomes $120$.  Taking the cost to be travel-time, Braess paradox is that adding a delay-free travel link can actually increase congestion and the users' travel time.

\begin{figure}[h]
\begin{center}
\includegraphics[height=0.35\columnwidth]{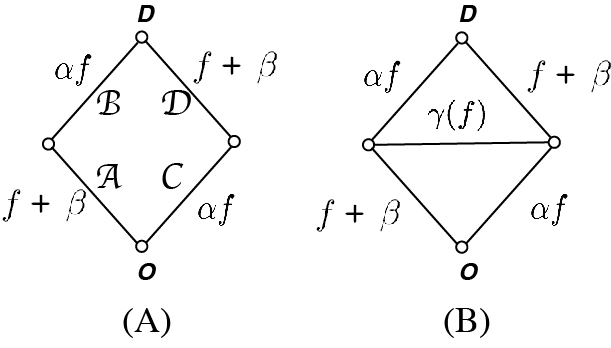}
\end{center}
\caption{The classic Braess paradox of congested network flows.}
\label{fig:jb:Braess}       
\end{figure}

Braess' paradox in this setting involves selfish social choices.  There is a similar apparent paradox in electric circuits, as noted by Cohen and Horowitz (\cite{Cohen}) that is a consequence of the laws of electrophysics.   In Fig.\ \ref{fig:jb:Circuit}, the question immediately arises as to how the horizontal connection changes the characteristics of the circuit on the top and bottom links.  Specifically, suppose that $R_2>R_1$.  Elementary invocation of Kirchhoff's current and voltage laws indicates that in the absence of the horizontal link, the currents $i_1$ and $i_2$ will be equal, but if the link is added with a moderate value of the resistance $R_3$, we will have $i_1>i_2$.  This change is not surprising and is consistent with the observation in \cite{Cohen} that introducing the path changes the voltage drop across the circuit.  It is also consistent with the observations in \cite{Blumsack1},\cite{Blumsack2} that adding this link may worsen congestion in similar models of power grid interconnections.  The interesting question posed in \cite{Blumsack1} and \cite{Blumsack2} is whether (and when) there is a useful tradeoff in grid design that balances the increased reliability of larger numbers of power lines against the congestion that may occur due to adding these lines to the network.  In what follows, we study a related question of when congestion may unexpectedly occur in a distribution network under the PDLC protocols described in \cite{Zhang2012}--\hspace{-1pt}\cite{Zhang2014b}.

\begin{figure}[h]
\begin{center}
\includegraphics[width=0.4\columnwidth]{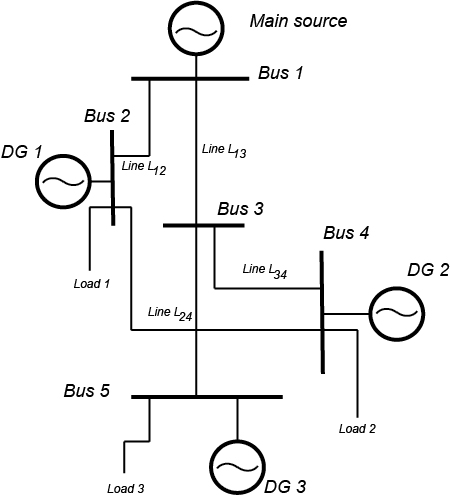}
\end{center}
\caption{Distribution grids of the future will feature increased amounts of local generation involving power sources with widely varying capacities and increased operating uncertainties. (Figure from Arghandeh {\em et al.} \cite{Arghandeh}).}
\label{fig:jb:Distribution}       
\end{figure}

Historically, distribution networks have had radial topologies, with distribution lines leading from a single trunk or power source to various commercial and residential loads.  This standard topology is likely to change with the ever increasing penetration of distributed generation in the form of wind, solar, plug-in electric vehicles, and many other forms of alternative energy.  Microgrids as depicted in Fig.\ \ref{fig:jb:Distribution} will become prevalent, and these will increasingly resemble miniaturized transmission networks in which sources and loads will be connected through a multiplicity of line links that can be opened and closed as needed to maintain the needed balance between capacity and demand.  It is against this backdrop that we examine the question of how the increased penetration of spatially distributed generation together with new electricity market models aimed at managing demand response will challenge the operating security of microgrid distribution networks.

\section{The effect of DG intermittency on congestion sensitivity to small changes in load}

Recent work has demonstrated how various communication protocols can be effectively employed in networks of smart microgrids to improve energy efficiency, decrease demand volatility, and ensure customer satisfaction---\cite{Zhang2012},\cite{Zhang2013a},\cite{Zhang2013b},\cite{Zhang2014a},\cite{Zhang2014b}.  Although this work has examined packet switched energy delivery in terms of temporal uncertainty in the operation of microgrids, the effects of spatial and network topology variability and uncertainty of both demand and DG capacity is not well understood.  While the precise magnitude of costs associated with mitigating the uncertainties of renewable generation sources  is not known, some estimates suggest that higher reserve margins will be required.  For instance, the historical averages of reserve requirements in the power grid point to 7\% to 8\% as generally sufficient to handle contingencies.  There are now predictions that if renewable penetration gets to the 33\% level (still a long way off) these requirements may go as high as 15\%, \cite{LaTimes}.

\begin{figure}[h]
\begin{center}
\includegraphics[width=0.35\columnwidth]{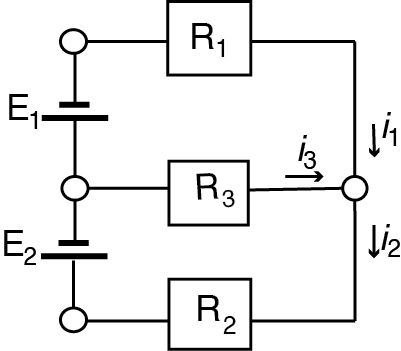}
\end{center}
\caption{A simple circuit with DC-voltage sources and resistive loads.}
\label{fig:jb:Circuit}       
\end{figure}

The implications of renewables for transmission networks remains a work in progress.  With the increasing likelihood that distribution networks will incorporate small-scale distributed generation as depicted in Fig.\ \ref{fig:jb:Distribution}, we turn to the question of how these intermittent sources will affect network congestion on a small scale.  Our working definition of congestion will be in terms of Fig.\ \ref{fig:jb:Circuit}, where we say that there is congestion if there is a significant difference between currents $i_1$ and $i_2$ and more importantly, between the energy losses due to heat $i_1^2R_1$ and $i_2^2R_2$.  We note that if  there is no horizontal link (i.e. if $i_3=0$ or equivalently, if $R_3\sim\infty$), then $i_1=i_2$, and if $R_1=R_2$, the heat losses at each resistor are equal as well.  These conclusions remain true irrespective of the magnitudes of the voltage sources.  If  $R_1=R_2$ and the voltage sources happen to be equal in magnitude, then $i_3=0$ no matter what value is assigned to $R_3$.  If there is any imbalance in the voltages $E_1$ and $E_2$, or if one of them is zero (think of a wind turbine or solar array being out of service due to weather conditions), then a small value of $R_3$ on the cross link of Fig. \ref{fig:jb:Circuit} can produce a very large difference in the currents $i_1$ and $i_2$.  We summarize this in the following:

\begin{proposition}
Suppose ${\rm R_1}={\rm R _2}={\rm R}$ and that ${\rm E_1}=0$.  Then if the cross link is not connected in Fig.\ \ref{fig:jb:Circuit} we have that $i_1=i_2={\rm E_2}/2{\rm R}$ while if the cross {\em is} connected and ${\rm R_3}$ is positive but small, we have $0\sim i_1\ll i_2\sim {\rm E_2}/{\rm R}$.  Moreover, the total heat loss across across the entire circuit is approximately twice the loss if the cross link is disconnected.
\end{proposition}

\begin{proof}
A simple application of Kirchhoff's circuit laws shows that the currents $i_j$ are related to the circuit voltages and resistances by relationships that are linear in the voltages but subtly nonlinear in the resistances.
\begin{equation}
\left(
\begin{array}{c}
i_1\\ i_2\\ i_3 \\
\end{array}
\right) = 
\left(\begin{array}{c}
\frac{{\rm E_1}({\rm R_2+R_3) + E_2R_3}}{{\rm R_1R_2+R_1R_3+R_2R_3}}\\[0.1in]
\frac{{\rm E_1R_3 + E_2(R_1+R_3)}}{{\rm R_1R_2+R_1R_3+R_2R_3}}\\[0.1in]
\frac{{\rm -E_1R_2 + E_2R_1}}{{\rm R_1R_2+R_1R_3+R_2R_3}}\\
 \end{array}
\right)
\label{eq:jb:Kirchhoff}
\end{equation}

The proof follows by setting ${\rm R_1=R_2=R}$ and considering the equation for small positive values of ${\rm R_3}$.
\end{proof}

The dramatic effect on current flow produced by connecting the cross-link in Fig.\ \ref{fig:jb:Circuit} is illustrated in Fig.\ \ref{fig:jb:CircuitPlotCurrent}.  Connecting the link with a very small resistance $R_3$ brings the current $i_1$ to essentially zero, while doubling the current $i_2$.  It is not surprising that the losses on the line through which $i_2$ flows increase dramatically.  It is also the case that the total losses for the circuit double when the connection is made.

\begin{figure}[h]
\begin{center}
\includegraphics[width=0.6\columnwidth]{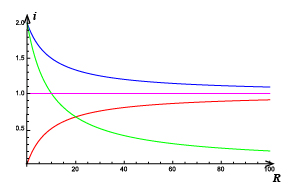}
\end{center}
\caption{The effect of connecting the cross link in Fig.\ \ref{fig:jb:Circuit}.  The magenta line indicates the per unit current values $i_1$,$i_2$ when the cross is not present.  The red, blue, and green curves are the values of the currents $i_1$,$i_2$, and $i_3$ (as functions of the resistance $R_3$) when the cross link is connected.  $R_1=R_2=R$ in all cases.}
\label{fig:jb:CircuitPlotCurrent}       
\end{figure}

\begin{definition} The {\it Loss Cost of the Link} (LCL) is defined as the ratio between the total network losses before and after the addition of a new link or other capacity enhancement.
Denoting the system losses before and after the capacity enhancement as $Loss'$ and $Loss$, respectively. The $LCL$ is defined as
\begin{equation}
LCL= \frac{Loss}{Loss'}.
\end{equation}
\end{definition}

The LCL concept is our version of the {\em price of anarchy} that is discussed in routing congestion problems in transportation networks, \cite{Steinberg}.  
Proposition 1 shows that the LCL for the balanced load network with $R_1=R_2$ is 2. The following Corollary extends the result to imbalanced networks of parallel connected resistors.

\begin{corollary} Suppose $E_1=0$, if the cross link is connected and $R_3$ is positive but small, then the price of anarchy of the Kirchhoff-Braess network is $R_1/R_2 + 1.$
\end{corollary}

\begin{proof}
Based on Kirchhoff's circuit Law, the current $i_j$ that flows through the resistor $R_j,j=1,2,3$ are
\begin{eqnarray}
i_1 = \frac{E_2 R_3}{R_1R_2+R_2R_3+R_3R_1},\nonumber\\
i_2 = \frac{E_2 (R_1+R_3)}{R_1R_2+R_2R_3+R_3R_1},\nonumber\\
i_3 = \frac{E_2 R_1}{R_1R_2+R_2R_3+R_3R_1}.\nonumber
\end{eqnarray}
The total system loss for small value of $R_3 \ll R_1,R_2$ is
\begin{equation}
Loss = \sum_{j=1}^{3} i_j^2R_j \approx \frac{E_2^2}{R_2}.
\end{equation}
On the other hand, the system loss before adding $R_3$ is 
\begin{equation}
Loss' = \frac{E_2^2}{(R_1+R_2)}.
\end{equation}
From the above two equations, we have the LCL is 
\begin{equation}
LCL = \frac{Loss}{Loss'} = \frac{R_1}{R_2} + 1 \geq 1.
\end{equation}
This ends the proof of Proposition 1. 
\end{proof}

A general scenario for Proposition 1 is to consider that two voltage sources are connected in the opposite direction at the same time, namely both $E_1$ and $E_2$ are non-zero. In typical distribution networks, the level of voltage is the same, say 110V at the user level. Therefore a much more complex electric network with the same level of voltage sources can be equivalently transformed into the network of two voltage sources connected in the opposite direction with a series connected to an equivalent resistor. In the following proposition, we show that the LCL will increase as the imbalance increases between the two resistors of the equivalent circuit.

\begin{proposition} Suppose that $E_1 = E_2 = E$. The LCL of Def.\ 1 will increase monotonically as $|R_1-R_2|$ increases. 
\end{proposition}

\begin{proof}
According to Kirchhoff's circuit Law, the current $i_j$ flows through the resistors are
\begin{eqnarray}
i_1 = \frac{E (R_2+2R_3)}{R_1R_2+R_2R_3+R_3R_1},\nonumber\\
i_2 = \frac{E (R_1+2R_3)}{R_1R_2+R_2R_3+R_3R_1},\nonumber\\
i_3 = \frac{E (R_1-R_2)}{R_1R_2+R_2R_3+R_3R_1}.\nonumber
\end{eqnarray}
The total system loss is 
\begin{equation}
Loss = i_1^2R_1 +  i_2^2R_2+ i_3^2R_3.
\end{equation}
Since the total system loss before connecting the resistor $R_3$ is
\begin{equation}
Loss' = \frac{(2E)^2}{R_1+R_2},
\end{equation}
we will have the LCL as follows
\begin{equation}
\frac{Loss}{Loss'} = \frac{(R_1-R_2)^2}{4(R_1R_2+R_2R_3+R_3R_1)}+1\geq 1,
\label{eq:jb:LCL}
\end{equation}
where equality occurs if and only if $R_1=R_2$. 
For an unbalanced resistor network (i.e.\ $R_1\ne R_2$), rewrite the expression for LCL in (\ref{eq:jb:LCL}) by introducing the positive variable $h=|R_1-R_2|$.  If we rewrite (\ref{eq:jb:LCL}) in terms of $h,R_1,R_3$ or $h,R_2,R_3$, a straightforward argument using elementary calculus and counting cases shows that for any $R_1,R_2,R_3>0,\ R_1\ne R_2$, LCL is an increasing function of $h$
\end{proof}


In decision theory, it is frequently desirable to base decisions on criteria that are known (or can be proven) to be monotonic in the decision variables.  The previous section showed that if there is a lack of balance in the voltage source distribution that large current imbalances (congestion) can occur.  In this section, we note that if there are imbalances in the resistances in our model network, then there will be a non-monotonic dependence of the losses $i_1^2R_1$ and $i_2^2R_2$ on the magnitude of the resistance $R_3$.

\begin{proposition}
Referring to the circuit of Fig.\ \ref{fig:jb:Circuit}, suppose $R_1<R_2$.  If $E_1=E_2$, then there is a non-monotonic dependence of the losses $i_1^2R_1$ and $i_2^2R_2$ on the cross-link resistance $R_3$.  Specifically, there is a critical value $R_3^{cr}$ such that for $R_3<R_3^{cr}$, $i_1^2R_1<i_2^2R_2$, while for $R_3>R_3^{cr}$,  $i_1^2R_1>i_2^2R_2$.
\end{proposition}

\begin{proof}
We compare
\[
i_1^2R_1 =\left(E_2R_3+E_1(R_2+R_3)\right)^2 R_1/D_e \ \ {\rm with} 
\]
\[
i_2^2R_2 =\left(E_1R_3+E_2(R_1+R_3)\right)^2 R_2/D_e 
\]
where $D_e=(R_1R_2+R_1R_3+R_2R_3)^2$.  The loss $i_1^2R_1$ will be larger than $i_2^2R_2$ precisely when the numerators of these expressions have the same magnitude relationships.  Recalling $E_1=E_2=({\rm say})\ E$, we have
\begin{equation}
i_1^2R_1\sim R_1E^2\left(R_2^2+4 R_2R_3+4 R_3^2\right)
\label{eq:jb:LHS}
\end{equation}
and 
\begin{equation}
i_2^2R_2\sim R_2E^2\left(R_1^2+4 R_1R_3+4 R_3^2\right)
\label{eq:jb:RHS}
\end{equation}
Obviously, for small $R_3>0$, the expression (\ref{eq:jb:LHS}) is greater than the expression (\ref{eq:jb:RHS}), but as $R_3$ becomes larger, the terms that are quadratic in $R_3$ dominate---making the expression (\ref{eq:jb:RHS}) the larger.
\end{proof}

\section{Optimal power-flow sensitivity to small changes in network parameters}

The kinds of sensitivity illustrated in the resistive load circuits in the preceding sections may be found as well in optimal power flow.  We revisit Example 6.16, pp. 252-254 in \cite{Gomez}.

\begin{example}
Consider the three-node network of Fig.\ \ref{fig:jb:PowerFlow}.  We examine a DC power flow model of the three bus network in which bus 1 and bus 2 are generators, while bus 3 is a load.  The production costs of operating the generators are $C_j(P_j)$ for $j=1,2$, where $P_j$ is the nodal power injection at the $j$-th bus.  These costs of generation are convex functions, reflecting the fact that as the power increases, the incremental cost rises superlinearly due to wear and tear on the machinery, decreased efficiency margins. The ``elastic price'' load at bus 3 is $P_D=P_C+P_E$, where $P_C$ is the inelastic component of the load and $P_E$ is the ``price elastic'' component of the load.  The line inductive reactances are $x_{1,2}$, $x_{13}$, $x_{23}$.  The power flow $P_{ij}$ on line $ij$ is given by
\begin{equation}
\frac{\theta_i-\theta_j}{x_{ij}},
\label{eq:jb:LineLoads}
\end{equation}
where $\theta_j$ is the power phase angle at the $j$-bus.  The nodal power injections are related to the power phase angles by the conductance matrix:
\begin{equation}
\left(\begin{array}{c}
P_1 \\
P_2\\
P_3
\end{array}
\right)  = B 
\left(
\begin{array}{c}
\theta_1 \\
\theta_2 \\
\theta_3
\end{array}
\right),
\label{eq:jb:conductance}
\end{equation}
where
\[
B = 
\left(
\begin{array}{ccc}
-b_{12}-b_{13} & b_{12} & b_{13}\\
b_{12} & -b_{12}-b_{23} & b_{23}\\
b_{13} & b_{23} & -b_{13}-b_{23} 
\end{array}\right).
\]
and where the line conductances $b_{ij}$ are the negative reciprocals of the line reactances, i.e.\ $-\frac{1}{x_{ij}}$.  The nodal power injections always sum to zero, as do the columns and rows of $B$, and since the power flow equations are invariant under a common phase shift of the $\theta_j$'s it is convenient to choose a reference bus (say bus 1) at which we set the phase angle $=0$.

The {\em optimal power flow} problem is to determine the nodal power injections $P_1$ and $P_2$ at the generator buses and the line flows $P_{ij}$ that optimize an objective function that accounts for generation costs $C_1(P_1)$ and $C_2(P_2)$ along with a consumer welfare cost that in the simplest formulation is evaluated only in terms of the price elastic  load at bus 3, $C_W(P_E)$.  The objective function to be minimized is written as $C_1(P_1)+C_2(P_2)-C_W(P_E)$.  Minimization is subject to constraints that the power flow solution does not exceed the rated capacities of the lines or buses.  Thus, feasible solutions must satisfy 
\begin{equation}
0\le P_j\le P_j^{max} \ \ \ {\rm and}\ \ 0\le P_{ij}\le P_{ij}^{max}.
\label{eq:jb:Constraints}
\end{equation}
In a power network, {\em congestion} is said to occur if the scheduled or desired power flow exceeds the rated capacity of either one or more of the lines or one or more of the generator buses.  

\begin{figure}[h]
\begin{center}
\includegraphics[width=0.75\columnwidth]{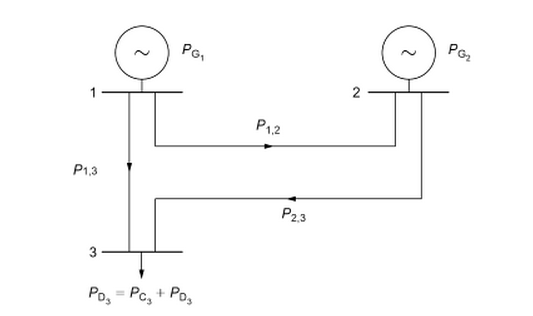}
\end{center}
\caption{From Gomez-Exposito {\em et al.}, \cite{Gomez}.}
\label{fig:jb:PowerFlow}       
\end{figure}

\begin{proposition}
For the three-bus network depicted in Fig.\ \ref{fig:jb:PowerFlow}, if the power flow solution (the phases in $\theta_j$ (\ref{eq:jb:conductance})) that minimizes the objective function
\begin{equation}
C_1(P_1)+C_2(P_2)-C_W(P_E)
\label{eq:jb:Objective}
\end{equation}
does not result in congestion---i.e.\ if none of the constraints (\ref{eq:jb:Constraints}) holds with equality---then the optimal value of (\ref{eq:jb:Objective}) is independent of the  conductance matrix $B$.
\end{proposition}

\begin{proof}
The power injections are related to the power phase angles by
\begin{equation}
\left(\begin{array}{c}
P_1\\
P_2\end{array}
\right) =
B_r
\left(
\begin{array}{c}
\theta_2 \\
\theta_3
\end{array}
\right),
\label{eq:jb:PowerPhase}
\end{equation}
where
\[
B_r =
\left(
\begin{array}{cc}
b_{12}  &b_{13} \\
-b_{12} - b_{23} & b_{23}
\end{array}
\right)
\]
is the reduced conductance matrix, and 
where we have assumed without loss of generality that the power phase angle at bus 1 is zero.  If no two of the conductances $b_{ij}$ are zero, then $B_r$ is nonsingular.  Using the invertible relationship (\ref{eq:jb:PowerPhase}), we may rewrite the objective function as
\[
f(\theta_2,\theta_3) = \hat C_1(\theta_2,\theta_3)+\hat C_2(\theta_2,\theta_3)-\hat C_W(\theta_2,\theta_3).
\]
Solving the critical point equations $\partial f/\partial\theta_j=0$ for $j=2,3$ leads to a minimizing solution, and this determines the optimal $(P_1,P_2)$ via (\ref{eq:jb:PowerPhase}).  This solution must be the same as what would have been obtained by minimizing (\ref{eq:jb:Objective}) directly.

\end{proof}

\begin{remark}
It holds much more generally that for a connected power grid, the optimal power injections from generators is independent of the line conductances {\em provided there is no congestion}.  The proof is a direct extension of the above but is omitted due to space limitations.  We instead examine the sensitivity of congestion to generation cost in the simple example at hand.\end{remark}

Adopting the cost functions of \cite{Gomez}, $C_j(P_j)=\beta_j P_j^2$ and $C_W(P_E)=\alpha P_E$, and recalling that $P_E=P_1+P_2-P_C$, the optimal power injections are easily seen to satisfy
\begin{equation}
P_j=\frac{\alpha}{2\beta_j},\ \ j=1,2,\ \ {\rm and}\ \ P_E=\frac12 \left(\frac{\alpha}{\beta_1} + \frac{\alpha}{\beta_2} - 2P_C\right).
\label{eq:jb:PowerInjection}
\end{equation}
Clearly, the cost coefficients $\alpha,\beta_1,\beta_2$ must be such that  the power injections are within the ranges (\ref{eq:jb:Constraints})---specifically the marginal value of consumer preference for load price elasticity must be in balance with the marginal costs of generation.  To evaluate the line loading produced by the power injections (\ref{eq:jb:PowerInjection}), we solve (\ref{eq:jb:PowerPhase}) and use (\ref{eq:jb:LineLoads}).  This yields

\begin{equation}
\begin{array}{ccl}
P_{12} & = &\frac{ \alpha}{ 2 \beta_1} \frac{b_{12} b_{23}}{D} -  \frac{\alpha} {2 \beta_2} \frac{ b_{12} b_{13} }{D}\\[0.1in] 
P_{13} & = & \frac{ \alpha}{ 2 \beta_1} \frac{ (b_{12}+b_{23}) b_{13} }{ D}  
   + \frac{ \alpha}{ 2 \beta_2}\frac{ b_{12} b_{13}}{ D} \\[0.1in]
 P_{23} & = &\frac{\alpha}{2 \beta_1 } \frac{ b_{12} b_{23}}{D }
  +  \frac{\alpha}{2 \beta_2}\frac{(b_{12} + b_{13}) b_{23}}{D},
\end{array}
\end{equation}
where
\[
D= b_{12} b_{13}+b_{23} b_{13}+b_{12} b_{23}.
\]
\end{example}

\vspace{0.15in}
It is expected that power grids will exhibit  the same kinds if sensitivity to changes in network topology and operating parameters that were noted in Sections III and IV.  
For the generation cost values considered in Example 6.16 in \cite{Gomez} ($\beta_1=1,\beta_2=1.675$), the explicit form of the uncongested optimum power injection at generator nodes 1 and 2 favors power produced by the cheaper generator (generator 1), although it is never the case that the DC load flow results in zero power being injected at bus 2.  The line loading between the less costly generator and the load ($P_{13}$) turns out to be a monotonically increasing function of the conductance $b_{13}$.  It is interesting to note that if $b_{13}$ is small enough in relation to $b_{23}$, the line loading will have $P_{23}>P_{13}$. and the load will draw more power from the line to the more expensive generator $G_2$.



Write the phase and line-loading relationship in matrix form: $P_{line}= H\cdot (\theta_2,\theta_3)^T$, where $P_{line}=(P_{12},P_{13},P_{23})^T$ and $H$ is the matrix representation specified by (\ref{eq:jb:LineLoads}):

\[
H=\left(\begin{array}{cc}
b_{12} & 0\cr
0& b_{13}\cr
-b_{23} & b_{23}
\end{array}\right).
\]
We can then express the line loadings directly in terms of the power injections by writing
\[
\left(\begin{array}{c}
P_{12}\cr P_{13} \cr P_{23} \end{array} \right) = H\;B_r^{-1} \left(\begin{array}{c}
P_1\cr P_2
\end{array}\right).
\]
For small values of $|b_{13}|$, this relationship is
\[
\left(\begin{array}{c}
P_{12}\cr P_{13} \cr P_{23} \end{array} \right) =\left(\begin{array}{cc}
1+  \epsilon & \epsilon \cr
 \epsilon & \epsilon \cr
 1+  \epsilon & 1+  \epsilon \end{array}\right)
 \left(\begin{array}{c}
P_1\cr P_2
\end{array}\right),
\]
where $\epsilon={\cal O}(b_{13})$.  With the power injected by generator $G_1$ being shifted from line $(1,3)$ to lines $(1,2)$ and $(2,3)$, it is reasonable to expect that congestion on these lines will be sensitive to changes in $P_1$ and generation cost parameter $\beta_1$.  Indeed a straightforward calculation shows that there is extreme sensitivity to the cost parameter with
\[
P_{23} = \frac{C}{\beta_1} + f(|b_{13}|),
\]
where $C$ is a positive constant, and $f$ is a smooth function of $|b_{13}|$.  

Once an uncongested optimum power flow lies outside the operating range of any component (\ref{eq:jb:Constraints}), the operating limit of that component becomes a binding constraint in terms of which the optimal power flow problem must be resolved.  (See \cite{Gomez}).  Rather than pursuing constrained optimal power flow at this point, we briefly explore the pervasiveness of the Kirchhoff-Braess phenomena in larger networks.

\section{The Case of Large Networks}

We shall consider the effect of attaching an arbitrary two-port voltage controlled circuit (e.g.\ a single resistor or single voltage source in the simplest cases) to any two points of an existing voltage controlled circuit of arbitrary topology.  It will be shown that the LCL resulting from the attachment will be $\ge 1$ in all cases.  We begin by recalling that a voltage controlled DC circuit is made up of resistors, capacitors, inductors, and voltage sources.  We have the following: 
\begin{definition} 
Points in the circuit at which two or more circuit elements are connected are called {\em nodes}.
\end{definition}

Fig. 1 shows a simple circuit with 5 nodes. For a DC circuit that is comprised purely of voltage sources and resistors, we can first choose one node to be the reference, usually called ground node. Then Kirchhoff's Law allows for the calculation of the voltages at each node of the circuit, relative to the reference ground node. Once all the node voltages are known, all currents in the circuit can be determined easily.

\begin{figure}[h]
\begin{center}
\includegraphics[width=0.55\columnwidth]{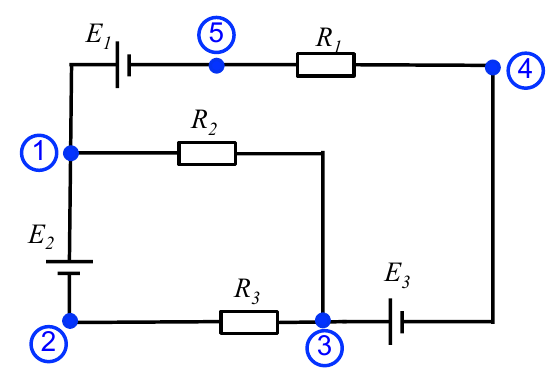}
\end{center}
\caption{A circuit with 5 nodes.}
\label{fig:node}   
\end{figure}

\begin{definition}For the graph of a DC circuit comprised purely of voltage sources and resistors, a {\em fundamental node basis} is a maximal set of nodes among which there exist no paths comprised purely of voltage source links. 
\end{definition}

A fundamental node basis may be formed as follows. We first find the set of all nodes $V$ where two or more circuit elements meet, and then create a sub-graph of the original circuit $G_v=(V,E_v)$ that  is comprised of the entire set of nodes $V$ and all voltage source links.   If  $G_v=(V,E_v)$ is a connected graph, any single node constitutes a fundamental node basis.  In this case, the LCL of the circuit resulting from connecting the external two-port circuit will be equal to $1$ because the voltage drop between any two nodes in the original circuit will be unaffected.

If $G_v=(V,E_v)$ is disconnected, we then can find all its connected components (with some components possibly consisting of single nodes). Then a fundamental node basis may be formed from the set of connected components of the given graph, by arbitrarily selecting one node from each connected component.  Since there exists at least one path comprised purely by voltage source links between every node pair in a connected component, we can easily get the voltage difference between the fundamental node and other nodes in that connected component by computing the algebraic sum of voltage sources on their  connecting paths. Thus once the fundamental node voltages are known, all node voltages of the circuit can be determined easily without solving the Kirchhoff's equations.
\begin{proposition}
\emph{If a DC circuit is comprised purely of voltage sources, resistors, capacitors and/or inductors, then the steady state LCL resulting from ADDING a resistance link must be $\geq$ 1 no matter what topological structure the original circuit has.}
\label{prop:jb:resistance}
\end{proposition}
\begin{proof} We start with the simplification of the circuit. If a DC circuit comprised purely of voltage sources, resistors, capacitors and/or inductors is in its steady state, then we can remove all loops with one or more capacitors (since capacitors act identically as open circuits in DC steady state) and replace all inductors with short circuits (since inductors act identically as short circuits in DC steady state) while keeping the total loss of the circuit unchanged.

We then can find all nodes where two or more circuit elements meet (only resistors or voltage sources are left after the above simplification), and mark them as $V=\{1,...,N\}$, where $N\geq2$ is the total number of nodes. Then each link between any pair of adjacent nodes \{$i,j$\} $(i,j=1,...,N,$ and $i\neq j)$ will only have either one voltage source or one resistor.

We next create the voltage source sub-graph of the original circuit $G_v=(V,E_v)$ which is comprised of the whole set of nodes $V=\{1,...,N\}$ and all voltage source links. The case that the voltage source sub-graph $G_v=(V,E_v)$ is a connected graph has been seen to be trivial and thus is ignored here. Suppose $G_v=(V,E_v)$ is disconnected, we then can find all its connected components and assume that they are $G_1=(V_1,E_1),...,G_M=(V_M,E_M)$ ($M$ is the total number of connected components). A fundamental node basis  then can be formed by arbitrarily selecting one node from each connected component, and we denote this basis by $V_F=\{v_1,...,v_M\}$ ($1 \leq v_1,...,v_M \leq N$). We denote the voltage at node $v_i$ by $e_{v_i}$ ($i=1,...,M$), and we choose $v_M$ as the reference ground node, i.e. $e_{v_M}=0$. 

It is easy to prove that the pair of endpoints of any resistor link in the original circuit  belong to either one connected component or two connected components defined above. For those resistor links whose pair of endpoints belong to one connected component, their power loss  will be unchanged no matter what kind of new link is introduced as the voltage drops between the pairs of endpoints are fixed. Thus we can use a scalar constant, say $P_1$, to denote the total loss of such resistor links.\\
Next we compute the loss of the resistor links whose pair of endpoints belong to two connected components. We denote the number of resistor links between  the pair of connected components $\{G_i,G_j\}$ $(i,j=1,...,M, and \ i\neq j$) by $L_{i,j}$, the current flowing away from the $i$-th connected component towards the $j$-th connected component on the $k$-th resistor link between $\{G_i,G_j\}$ $(k=1,...,L_{i,j}$) by $I_{i,j,k}$ ($I_{i,j,k}$=$-I_{j,i,k}$), and the resistor on the $k$-th link between $\{G_i,G_j\}$ by $R_{i,j,k}$ ($R_{i,j,k}$=$R_{j,i,k}>0$). Then the total loss of such resistor links is given by
\begin{equation}
P_2 =  \sum_{i=1}^{M-1}\sum_{j=i+1}^{M}\sum_{k=1}^{L_{i,j}}I_{i,j,k}^2R_{i,j,k}
\end{equation}\\
where $I_{i,j,k}$ can be expressed as
\begin{equation}
I_{i,j,k}=\frac{e_{v_i}+e_{P_{v_{i,k}}}-e_{v_j}-e_{P_{v_{j,k}}}}{R_{i,j,k}}.
\end{equation}
with $e_{P_{v_{i,k}}}$ denoting the  algebraic sum of voltage sources on the path connecting the fundamental node $v_i$ and the node to which the resistor $R_{i,j,k}$ is attached in the connected component $G_i$.  \ ($e_{P_{v_{i,k}}}=0$ if one endpoint of the resistor $R_{i,j,k}$ is directly connected to the fundamental node $v_i$.),  Similarly, $e_{P_{v_{j,k}}}$  is the sum of voltages along the path connecting the other endpoint of the resistor to the fundamental node $v_j$.  ($e_{P_{v_{j,k}}}=0$ if the other endpoint of the resistor $R_{i,j,k}$ is directly connected to the fundamental node $v_j$.)

A potential function whose physical meaning is the total loss of all resistors in the original circuit can be created based on the variables in the fundamental node set.

\noindent This is given explicitly by
\begin{equation}
\begin{split}
P & =  P_1+P_2\\
&=P_1+\sum_{i=1}^{M-1}\sum_{j=i+1}^{M}\sum_{k=1}^{L_{i,j}}\frac{(e_{v_i}+e_{P_{v_{i,k}}}-e_{v_j}-e_{P_{v_{j,k}}})^2}{R_{i,j,k}}
\end{split}
\label{eq:jb:loss}
\end{equation}
where $\{e_{v_1},...,e_{v_{M}}\}$  are the node voltages.  We assume $e_{v_M}=0$, and the value of other node voltages before and after adding a new link are $(\bar{e}_{v_1},...,\bar{e}_{v_{M-1}})$ and $(e_{v_1}^{'},e_{v_2}^{'},...,e_{v_{M-1}}^{'})$, respectively. We shall show that the potential function reaches its minimum loss level before adding a new link, i.e. $P(\bar{e}_{v_1},...,\bar{e}_{v_{M-1}})$ is always $\leq$ $P(e_{v_1}^{'},e_{v_2}^{'},...,e_{v_{M-1}}^{'})$.

We first show that the potential function is a strict convex function. Suppose that 
\begin{equation}
P_{i,j,k}=\frac{(e_{v_i}+e_{P_{v_{i,k}}}-e_{v_j}-e_{P_{v_{j,k}}})^2}{R_{i,j,k}} 
\end{equation}
 $P_{i,j,k}$ is obviously a convex function of the node voltages.  
 Since we assume $e_{v_M}=0$ and suppose
\begin{equation}
P_{i,M,k}=\frac{(e_{v_i}+e_{P_{v_{i,k}}}-e_{P_{v_{M,k}}})^2}{R_{i,M,k}}  \quad (i=1,...,M-1)
\end{equation}
then $P_{i,M,k}$ is also
strictly convex provided $R_{i,M,k}$is finite. Since
\begin{equation}
P=  P_1 + \sum_{i=1}^{M-1}\sum_{j=i+1}^{M}\sum_{k=1}^{L_{i,j}}P_{i,j,k}
\end{equation}
$P$ is also a strictly convex function, taking its global minimum where all its partial derivatives are zero.

We  show that all partial derivatives of potential function are zero when   $(e_{v_1},e_{v_2},...,e_{v_{M-1}}) = (\bar{e}_{v_1},...,\bar{e}_{v_{M-1}})$. 
Since 
\begin{equation}
\begin{split}
\frac{\partial P}{\partial e_{v_i}}=  & 2\times(\sum_{j=1}^{i-1}\sum_{k=1}^{L_{i,j}}\frac{e_{v_i}+e_{P_{v_{i,k}}}-e_{v_j}-e_{P_{v_{j,k}}}}{R_{i,j,k}}+ \\
& \sum_{j=i+1}^{M}\sum_{k=1}^{L_{i,j}}\frac{e_{v_i}+e_{P_{v_{i,k}}}-e_{v_j}-e_{P_{v_{j,k}}}}{R_{i,j,k}}), 
\end{split}\label{kir}
\end{equation} 
the partial derivative of $P$ in the direction $e_{v_i}$ is exactly double the algebraic sum of all currents flowing on the original resistors that meet at the $i$-th connected component.  Because the node voltages $(e_{v_1},e_{v_2},...,e_{v_{M-1}})=(\bar{e}_{v_1},...,\bar{e}_{v_{M-1}})$ must satisfy Kirchhoff's Current Law, the partial derivatives in Equ.~\ref{kir} must be zero for $i=1,\dots,M-1$.  The potential function thus reaches its global minimum loss level under normal operating conditions. We consider what happens if a new link (two-port circuit) is attached to any pair of nodes,  If both endpoints of the newly introduced link are added to the same connected component of $G_v(V,E_v)$, defined above, then the total loss of the original circuit will remain the same as $(\bar{e}_{v_1},...,\bar{e}_{v_{M-1}})=(e_{v_1}^{'},e_{v_2}^{'},...,e_{v_{M-1}}^{'})$. However, if the two endpoints of the new link are added to two different connected components, say $G_i$ and $G_j$, and if we denote the current flowing on the new link from $G_i$ to $G_j$ by $I_{i,j}^{new}$, then we will have 
\begin{equation}
\begin{split}
\frac{\partial P}{\partial e_{v_i}^{'}}&=2\times(\sum_{j=1}^{i-1}\sum_{k=1}^{L_{i,j}}\frac{e_{v_i}^{'}+e_{P_{v_{i,k}}}-e_{v_j}^{'}-e_{P_{v_{j,k}}}}{R_{i,j,k}}+ \\
&\ \ \ \sum_{j=i+1}^{M}\sum_{k=1}^{L_{i,j}}\frac{e_{v_i}^{'}+e_{P_{v_{i,k}}}-e_{v_j}^{'}-e_{P_{v_{j,k}}}}{R_{i,j,k}})\\&=2\times I_{i,j}^{new}.
\end{split}
\end{equation}  
After adding the new link, the algebraic sum of all currents flowing on the original resistors that meet at the $i$-th connected component is determined by the current $I_{i,j}^{new}$ which is not necessarily zero as there may be current import and export to the newly added link from the original system. Therefore the original loss equilibrium is perturbed away from its minimum level, and thus the system loss of the original system will have increased.  This proves the proposition.
\end{proof}

\begin{remark}
The multi-node connected components of $G_v(V,E_v)$ may have arbitrarily complex topologies--including tree and loop components.  If a fundamental node basis that differs from the one chosen to define the loss in (\ref{eq:jb:loss}) is chosen, the form of the loss function (\ref{eq:jb:loss}) will differ accordingly.  The critical point determined by setting the partial derivatives in (\ref{kir}) equal to zero will minimize the new expression for loss.  It follows from the Kirchhoff circuit laws, that the minimizing values in both representations are the same---as we would expect.
\end{remark}

\begin{remark}
A more general version of Proposition \ref{prop:jb:resistance} can be established.  Indeed, the proof as given applies to the connection any two-port voltage controlled circuit to an existing voltage controlled circuit in steady state.  Although the proof becomes more involved, a similar result holds for the addition of an $n$-port external circuit.
\end{remark}

\section{Conclusions}

It has long been recognized that contingencies like the loss of a major power line can pose significant threats to the secure operation of the power grid. The focus of this paper has been on the way that seemingly small changes can have large effects. We have presented examples of simple circuits and networks that display the kinds of sensitivity to small changes in operating parameters that will need to be better understood as smart microgrids become an increasingly important part of power distribution networks. Secure operation of these microgrids will require the real-time coordinated control of increasing numbers of small-scale generation resources and consumer-provided demand response assets while respecting the safe operating ranges of all lines and equipment. Taking inspiration from concepts of congestion in traffic networks, we have studied what we call Kirchhoff-Braess phenomena---the apparent worsening of congestion due to the addition of a lightly loaded line connecting points in the network. We have defined a power network analogue of the price of anarchy that we call the {\em Loss Cost of the Link} (LCL). This is simply the ratio of losses after and before the asset (e.g.\ a line) has been added. It has been shown that this ratio is frequently greater than one, and this has been studied in detail for simple small networks and in considerable generality for voltage controlled networks in Section VI. Space does not permit treatment of current-controlled circuits, but a corresponding theory characterizes situations in which network losses increase and other situations in which there are decreases in network losses when a new circuit link is added. An important goal of future research is an understanding of general classes of optimal power flows in which the cost sensitivity of network congestion as treated in Section V is extended to realistic transmission and distribution networks. Models that capture degradation of voltage, frequency and other important physical parameters are being developed. 


\section{Acknowledgments}
The authors thank colleagues Michael Caramanis and Justin Foster for bringing Braess paradox in distribution networks to their attention. 


\end{document}